\documentclass[letterpaper,11pt]{article}
\usepackage{amsmath}
\usepackage{amssymb} 
\usepackage{amsthm}
\usepackage{graphicx}
\usepackage[total={442.1558pt, 647pt},marginratio={1:1, 1:1}]{geometry}

\newtheorem{lemma}{Lemma}
\newtheorem{prop}[lemma]{Proposition}
\newtheorem*{maintheorem}{Main Theorem}
\newtheorem*{tarskitheorem}{Knaster--Tarski Theorem}
\newtheorem*{conjecture}{Conjecture}

\usepackage[font=small,margin=15pt]{caption}

\newcommand{\Nset}{\mathbb N}
\newcommand{\Rset}{\mathbb R}
\newcommand{\dist}{\mathrm{dist}}
\DeclareMathOperator{\dom}{dom}
\DeclareMathOperator{\bisect}{bisect}
\newcommand{\bd}{\partial}
\newcommand{\closure}{\overline}

\title{Distance $k$-sectors exist\footnotetext{%
A preliminary form
of the results was announced in 
Section~4 of \cite{zone_unique_eurocg}. 
For the case $k=3$,
some of the methods and results were found
essentially independently in \cite{reem09:_voronoi_and_zone_diagrams}.
}}

\author{%
Keiko Imai%
\thanks{Department of Information and System Engineering, 
        Chuo University. 
        \texttt{imai@ise.chuo-u.ac.jp}} 
\and 
Akitoshi Kawamura%
\thanks{Department of Computer Science, University of Toronto. 
        \texttt{kawamura@cs.toronto.edu}}
\and
Ji\v{r}\'{i} Matou\v{s}ek%
\thanks{Department of Applied Mathematics, Charles University. 
        \texttt{matousek@kam.mff.cuni.cz}}
\and
Daniel Reem%
\thanks{Department of Mathematics, 
        The Technion -- Israel Institute of Technology. 
        \texttt{dream@tx.technion.ac.il}}
\and
Takeshi Tokuyama%
\thanks{Graduate School of Information Sciences, Tohoku University. 
        \texttt{tokuyama@dais.is.tohoku.ac.jp}}}

\date{December 2009}

\begin{document}

{\tabcolsep15pt %% this splits the five authors into 3 and 2, not 4 and 1
\makeatletter\let\@fnsymbol\@arabic\makeatother
\maketitle
}

\begin{abstract}
The \emph{bisector} of two nonempty sets $P$ and $Q$
in $\Rset ^d$ is the set of all points with equal distance to $P$ and to $Q$. 
A \emph{distance $k$-sector} of $P$ and $Q$, 
where $k \geq 2$ is an integer, 
is a $(k-1)$-tuple $(C_1,C_2,\ldots,C_{k-1})$ such that
$C_i$ is the bisector of $C_{i-1}$ and $C_{i+1}$
for every $i=1,2,\ldots,k-1$, where $C_0=P$ and $C_k=Q$.
This notion, for  the case where $P$ and $Q$ are points
in $\Rset ^2$, was introduced by Asano, Matou\v{s}ek, and Tokuyama,
motivated by a question of Murata in VLSI design. They established
the existence and uniqueness of the distance trisector in this
special case. We prove the existence of a distance $k$-sector
for all $k$ and
for every two disjoint, nonempty, closed sets $P$ and $Q$ in 
Euclidean spaces of any (finite) dimension
(uniqueness remains open), 
or more generally, in proper geodesic spaces. 
The core of the proof is a new notion of \emph{$k$-gradation}
for $P$ and $Q$, whose existence (even in an arbitrary metric space)
is proved using the Knaster--Tarski fixed point theorem,
by a method introduced by Reem and Reich for a slightly different purpose. 
\end{abstract}

\section{Introduction}

The \emph{bisector} of two
nonempty sets $X$ and $Y$ in Euclidean space, 
or in an arbitrary metric space $(M, \dist)$, 
is defined as
\begin{equation}
 \label{equation: definition of bisector}
\bisect (X, Y) = \{\, z \in M : \dist (z, X) = \dist (z, Y) \,\}, 
\end{equation}
where $\dist (z, X) = \inf _{x \in X} \dist(z,x)$
denotes the distance of $z$ from  a set~$X$. 

Let $k\ge 2$ be an integer and let $P$, $Q$ be disjoint nonempty 
sets in $M$ 
called the \emph{sites}.
A \emph{distance $k$-sector} (or simply \emph{$k$-sector})
of $P$ and $Q$ is  
a $(k - 1)$-tuple $(C _1, \dots, C _{k - 1})$ of 
nonempty subsets of $M$ such that 
\begin{align}
\label{equation: definition of k-sector}
 C _i & = \bisect (C _{i - 1}, C _{i + 1}), 
&
 i = 1, \ldots, k - 1, 
\end{align}
where $C _0 = P$ and $C _k = Q$ (see Figures 
\ref{figure: k-sector example} and \ref{figure: 7-sector}). 

\begin{figure}
\begin{center}
\includegraphics[clip,scale=1.1]{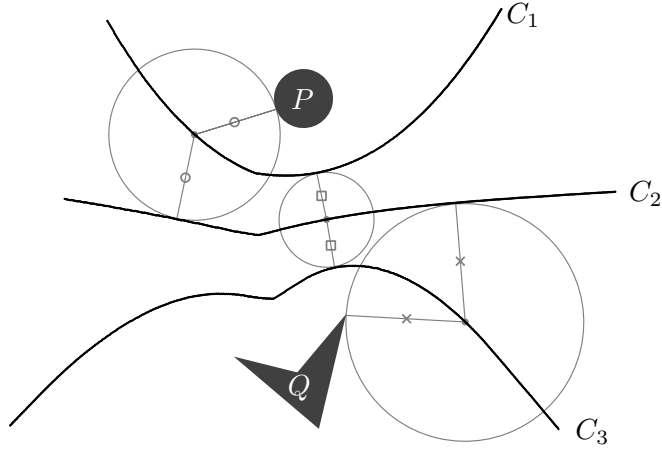}
\caption{A $4$-sector $(C _1, C _2, C _3)$ of sets $P$ and $Q$ in Euclidean plane.
Each point on the curve $C _i$ is at the same distance
from $C _{i - 1}$ and $C _{i + 1}$.
Note that $C _2$ is not the bisector of $P$ and $Q$.}
 \label{figure: k-sector example}
\end{center}
\end{figure}

Distance $k$-sectors were introduced by
Asano et al.~\cite{asano07:_distan_trisec_curve},
motivated by a question of Murata from VLSI design:
Suppose that we are given a topology of a circuit layer, and we need to 
put $k - 1$ wires through a corridor between given two sets of obstacles 
(modules and other wires) on the board. 
The circuit will have a  
high failure probability if the gaps between the wires are narrow. 
Which curves should the wires follow 
in order to minimize the failure probability? 
If $k = 2$, the curve should be the distance bisector; 
in general, each curve should be the bisector 
of its adjacent pair of curves, 
as stated in the definition of a $k$-sector. 

\begin{figure}
\begin{center}
\includegraphics[clip,scale=0.315]{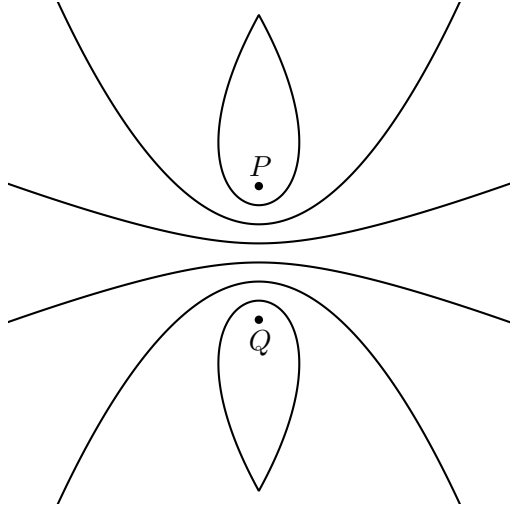}%
\smash{\makebox[0pt][r]{\raisebox{124pt}{$P$}\hspace*{90pt}}}%
\smash{\makebox[0pt][r]{\raisebox{58pt}{$Q$}\hspace*{90pt}}}%
\caption{A $7$-sector of two singleton sets $P$ and $Q$ in Euclidean plane.}
 \label{figure: 7-sector}
\end{center}
\end{figure}

\looseness=-1 % ad hoc
A similar problem occurs also in designing routes of $k - 1$ autonomous robots moving 
in a narrow polygonal corridor.  
Each robot has its own predetermined route 
(say, it is drawn on the floor with a coloured tape that the robot can recognize) 
and tries to follow it.
We want to design the routes to be far away from each other so that the robots can easily avoid collision. 

Despite its innocent definition, 
it is nontrivial to find a $k$-sector even in Euclidean plane. 
The bisector of two point sites $P$ and $Q$ in $\Rset ^2$ is a line, and   
an elementary geometric argument shows that there is a distance
$4$-sector of them consisting of a straight line and two parabolas.
However, the problem was not investigated for other values of $k$ 
until Asano et al.~\cite{asano07:_distan_trisec_curve} proved 
the existence and uniqueness of the $3$-sector of two points in Euclidean plane. 
Chun et al.~\cite{chun07:_distan_trisec_of_segmen_and}
extended this to the case where $Q$ is a line segment.

We give the first proof of existence of distance
$k$-sectors in Euclidean spaces for a general~$k$. 
This improves on the previous
proofs in generality and simplicity even for $k=3$. 

\begin{maintheorem}
Every two disjoint nonempty closed sets
$P$ and $Q$ in Euclidean space~$\Rset ^d$, or more generally, 
in a proper geodesic metric space, 
have at least one $k$-sector.
\end{maintheorem}

Here, a metric space $(M, \dist)$ is called \emph{proper} if 
all closed balls are compact. 
It is called \emph{geodesic}
if for every two distinct points $x$, $y\in M$ there is a \emph{metric segment}
in $M$ connecting them, i.e., an isometric mapping $\gamma \colon [a, b] \to M$
of an interval $[a,b]\subset \Rset$ with $\gamma(a)=x$ and $\gamma(b)=y$.
In particular, a convex subset of a normed space is a geodesic metric space. 
Another example is 
the surface of a sphere, 
where the distance between two points is measured by the length of the 
shortest path on the surface connecting them. 
Geodesic metric spaces are a reasonably general class of metric
spaces in which our arguments go through, although one could
probably make up even more general conditions.
 
Let us remark that if $\dist(P,Q)>0$ and $k=3$, 
then the properness assumption can be omitted; 
see \cite{reem09:_voronoi_and_zone_diagrams}
for a proof. 

On the other hand, $k$-sectors need not exist in arbitrary metric
spaces. A simple example for
$k=3$ is the subspace $M=\{-1,0,1\}$ of the real line,
$ P=\{1\}$, and $Q=\{-1\}$.
\medskip

From now on, unless otherwise noted,
subscripts $i$ and $j$ range over $1$, \ldots, $k - 1$;
for example, $(C _i) _i$ stands for the $k$-tuple $(C _1, \dots, C _{k - 1})$.

\paragraph{Gradations} One of the main steps in the proof of
the main theorem is introducing the notion of a \emph{$k$-gradation}
of $P$ and $Q$, which is related to a $k$-sector but easier
to work with. 
First, for nonempty sets $X$, $Y$ 
in a metric space $(M, \dist)$, 
we define the
\emph{dominance region} of $X$ over $Y$ by
\begin{equation}
 \dom (X, Y) = \{\, z \in M : \dist (z, X) \leq \dist (z, Y) \,\}.
\end{equation}
A \emph{$k$-gradation} between nonempty subsets $P$ and $Q$ of $M$
is a $(k - 1)$-tuple $(R _i, S _i) _i$ of
pairs of subsets of $M$ satisfying
\begin{align}
\label{equation: definition of gradation}
 R _i
&
=
 \dom (R _{i - 1}, S _{i + 1}),
&
 S _i
&
=
 \dom (S _{i + 1}, R _{i - 1}),
&
 i = 1, \ldots, k - 1,
\end{align}
where $R _0 = P$ and $S _k = Q$.

Using the Knaster--Tarski fixed point theorem~%
\cite{tarski55:_lattic_theor_fixpoin_theor_and_its_applic},
we prove in Section~\ref{section: main proof}
that $k$-gradations always exist:

\begin{prop}\label{p:gradations}
For every nonempty sets $P$ and $Q$
in an arbitrary metric space $(M,\dist)$, there exists
at least one $k$-gradation.
\end{prop}

The idea of applying the Knaster--Tarski theorem
to a similar setting is from 
\cite{reem09:_zone_and_doubl_zone_diagr},
where it is used to prove the existence of \emph{double zone diagrams}.
A slight modification of 
Proposition~\ref{p:gradations} also holds in the more general
setting of \emph{m-spaces}~\cite{reem09:_zone_and_doubl_zone_diagr}.

In Section~\ref{section: gradation},
we establish the following connection between $k$-gradations and
$k$-sectors.

\begin{prop}
\label{p:k-sector is a boundary}
Let $P$, $Q$ be nonempty, disjoint, closed sets
in a proper geodesic metric space. Then 
a $(k - 1)$-tuple $(C _i) _i$ of sets is
a $k$-sector of $P$ and $Q$ if and only if
\begin{align}
\label{equation: k-sector is a boundary}
 C _i & = R _i \cap S _i, & i = 1, \dots, k - 1
\end{align}
for some $k$-gradation $(R _i, S _i) _i$
between $P$ and $Q$.
\end{prop}
For instance, 
the $k$-sectors $(C _i) _i$ in Figures \ref{figure: k-sector example}, 
\ref{figure: 7-sector} and 
\ref{figure: l1 non-unique trisector} 
correspond to the $k$-gradations $(R _i, S _i) _i$ where
each $R _i$ is the union of $C _i$ and the region above it, 
and each $S _i$ is the union of $C _i$ and the region below it. 

The main theorem is an immediate consequence of Propositions~\ref{p:gradations}
and \ref{p:k-sector is a boundary}. 

\paragraph{\boldmath $3$-gradations and zone diagrams}
A \emph{zone diagram} of $P$, $Q$ is, according to the general
definition of Asano et al.~\cite{asano07:_zone_diagr},
a pair of sets $(A,B)$ such that $A=\dom(P,B)$ and $B=\dom(Q,A)$.
By comparing the definitions, we can see that
if $((R_1, S_1), (R_2,S _2))$ is a $3$-gradation
for $P$, $Q$, then $(R_1,S_2)$ is a zone diagram of $P$, $Q$.
Conversely, given a zone diagram $(A,B)$,
we can set $R_1:=A$, $S_2:=B$, $R_2:=\dom(R_1,Q)$,
$S_1:=\dom(S_2,P)$ to obtain a $3$-gradation
(we note that $R_2$ and $S_1$ are uniquely determined
by $R_1$ and $S_2$).

The existence of zone diagrams of arbitrary two nonempty sets
in an arbitrary metric space (and even in the still
more general setting of m-spaces) was proved
by Reem and Reich~%
\cite[Theorem~5.6]{reem09:_zone_and_doubl_zone_diagr}.
By the above, it immediately implies the existence of
$3$-gradations, a special case of Proposition~\ref{p:gradations}.

\paragraph{Uniqueness} Kawamura et al.~\cite{kmt-zone}
(also see \cite{zone_unique_eurocg} for a preliminary version)
proved the existence and uniqueness of zone diagrams in $\Rset ^d$
(for finitely many closed and pairwise separated sites) under
the Euclidean distance, and more generally, under
any smooth and uniformly convex norm.  
By Proposition~\ref{p:k-sector is a boundary}, 
this implies the uniqueness of trisectors under the same
conditions. This is the most general uniqueness result
for $k$-sectors we are aware of.

For general metrics, $k$-sectors need not be unique.
A simple example, for the $\ell_1$ metric in the plane
(given by $\dist(x,y)=|x_1-y_1|+|x_2-y_2|$), is shown
in Figure~\ref{figure: l1 non-unique trisector}; 
essentially, it was discovered by Asano and Kirkpatrick
\cite{asano06:_distan_trisec_curves}. The set $C_1$
is a polygonal curve, while $C_2$
is ``fat'', consisting of two straight segments
and two quadrants. A different trisector is obtained
as a mirror reflection of the one shown.

\begin{figure}
\begin{center}
\includegraphics[clip,scale=1.0]{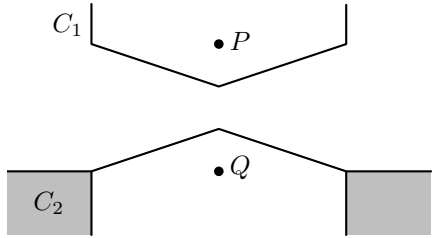}
\caption{A $3$-sector $(C _1, C _2)$ of $P$ and $Q$ 
under the $\ell_1$ norm.}
 \label{figure: l1 non-unique trisector}
\end{center}
\end{figure}

Thus, uniqueness of $k$-sectors requires some geometric assumptions
on the underlying metric space. We will further comment
on this issue in Section~\ref{section: constructive}.

\paragraph{\boldmath Construction of $k$-sectors}
Our  existence proof for $k$-sectors, based on the Knaster--Tarski
theorem, is somewhat nonconstructive. 
In Section~\ref{section: constructive}, we discuss a more constructive
approach, which re-establishes Proposition~\ref{p:gradations}
under stronger assumptions, but
which yields an iterative algorithm (in a similar
spirit as in \cite{asano07:_zone_diagr}). We have no rigorous
results about the speed of its convergence, but in practice
it has been used successfully for approximating
$k$-sectors and drawing pictures such as 
Figure~\ref{figure: k-sector example}. Such computations
also support our belief that $k$-sectors in
Euclidean spaces are unique, at least for two
point sites in the plane.

\section{\boldmath The existence of $k$-gradations}
\label{section: main proof}

Here we prove Proposition~\ref{p:gradations}.
A set $\mathcal L$ equipped with a partial order $\mathord\leq$ is called a 
\emph{complete lattice} 
if every subset $\mathcal D \subseteq \mathcal L$ has 
an infimum~$\bigwedge \mathcal D$ 
(the greatest $x \in \mathcal L$ such that 
$x \leq y$ for all $y \in \mathcal D$) and 
a supremum~$\bigvee \mathcal D$ 
(the least $x \in \mathcal L$ such that 
$x \geq y$ for all $y \in \mathcal D$). 
We say that a function $F \colon \mathcal L \to \mathcal L$ 
on a complete lattice~$\mathcal L$ is 
\emph{monotone}
if $x \leq y$ implies $F (x) \leq F (y) $. 

\begin{tarskitheorem}[\cite{tarski55:_lattic_theor_fixpoin_theor_and_its_applic}]
Every monotone function on a complete lattice has a fixed point. 
\end{tarskitheorem}

The proof of this theorem is simple: 
It is routine to verify that 
the least and the greatest fixed points of 
a monotone function $F \colon \mathcal L \to \mathcal L$ are given by 
\begin{align}
\label{equation: tarski fixed points}
 & \bigwedge \{\, x \in \mathcal L : x \geq F (x) \,\}, &
 & \bigvee \{\, x \in \mathcal L : x \leq F (x) \,\}, 
\end{align}
respectively. 

\begin{proof}[Proof of Proposition~\ref{p:gradations}.]
Let $\mathcal L$ be the set of all 
$(k - 1)$-tuples $(R _i, S _i) _i$ of 
pairs of subsets of the considered
metric space $M$ 
satisfying $R _i \supseteq P$, $S _i \supseteq Q$ and $R _i \cup S _i = M$. 
We
define the order~$\mathord\leq$ on $\mathcal L$ by setting 
$(R _i, S _i) _i \leq (R' _i, S' _i) _i$ if
$R _i \subseteq R' _i$ and $S _i \supseteq S' _i$ for all $i = 1, \ldots, k - 1$. 
It is easy to see that $\mathcal L$ with this order~$\mathord\leq$ 
is a complete lattice in which 
the infimum and supremum of $\mathcal D \subseteq \mathcal L$ are given by 
\begin{align}
\label{equation: infimum and supremum}
 \bigwedge \mathcal D 
&
=
 \Bigl(
  \bigcap _{(R _j, S _j) _j \in \mathcal D} R _i, 
  \bigcup _{(R _j, S _j) _j \in \mathcal D} S _i
 \Bigr) _i, 
&
 \bigvee \mathcal D 
&
=
 \Bigl(
  \bigcup _{(R _j, S _j) _j \in \mathcal D} R _i, 
  \bigcap _{(R _j, S _j) _j \in \mathcal D} S _i
 \Bigr) _i. 
\end{align}
We define $F \colon \mathcal L \to \mathcal L$ by 
\begin{equation}
\label{equation: definition of F}
 F \bigl( (R _i, S _i) _i \bigr) 
=
 \bigl( \dom (R _{i - 1}, S _{i + 1}), \dom (S _{i + 1}, R _{i - 1}) \bigr) _i, 
\end{equation}
where $R _0 = P$ and $S _k = Q$. 
It is easy to see that $F$ is well-defined and monotone. 
By the Knaster--Tarski Theorem, 
$F$ has a fixed point, 
which is a $k$-gradation by definition. 
\end{proof}

\section{\boldmath Dominance regions, $k$-gradations, and $k$-sectors}
\label{section: gradation}

The goal of this section is to  prove Proposition~\ref{p:k-sector is a boundary}. 
We write $\bd Z$ for the boundary of a closed set $Z$.
We begin with observing that, 
for arbitrary nonempty sets $X$, $Y$ in any metric space, 
the set $\bisect (X, Y) = \dom (X, Y) \cap \dom (Y, X)$
contains $\bd \dom (X, Y)$. 
Moreover, 
if the metric space is geodesic (and hence connected), then 
$\bisect (X, Y)$ is nonempty. 
For otherwise, $\dom (X, Y)$ and $\dom (Y, X)$ would be two disjoint closed sets 
covering the whole space. 

\begin{lemma}
\label{lemma: properties of dom}
Let $X$, $Y$, $Z$ be nonempty closed sets in a proper geodesic metric space.
Note that $D = \dom (X, Y)$ and $C = \bisect (X, Y)$ are nonempty. 
If $D$ and $Z$ are disjoint, then
\begin{enumerate}
\renewcommand{\theenumi}{\textup{(\alph{enumi})}}
\renewcommand{\labelenumi}{\theenumi}
\item \label{enumi: dominance region of the boundary}
$
  \dom (D, Z) = \dom (C, Z)
$, $
  \dom (Z, D) = \dom (Z, C)
$, 
\item \label{enumi: bisector of the boundary}
$
  \bisect (D, Z) = \bisect (C, Z)
$. 
\end{enumerate}
\end{lemma}

\begin{proof}
Part~\ref{enumi: bisector of the boundary} follows from 
\ref{enumi: dominance region of the boundary} 
using $\bisect(X,Y)=\dom(X,Y)\cap\dom(Y,X)$. 

To show \ref{enumi: dominance region of the boundary}, 
we claim that 
\begin{align}
\label{equation: dom* and bisect look the same from inside}
  \dist (a, Z) 
&
 >
  \dist (a, C) 
&
 \text{for all} \ a \in D. 
\end{align}
Indeed,  let $z \in Z$ be a point attaining the distance to $a$;
i.e., $\dist (a, z) = \dist (a, Z)$ (the distance is attained
since the intersection of $Z$ with the ball of radius
$2\dist(a,Z)$ around $a$ is compact).
There is a segment connecting $a$ and $z$---that is, 
a metric segment (see the definition following the Main Theorem); 
for $\Rset ^d$ this simply means a line segment. 
The segment is a connected set 
containing both $a \in D$ and $z \notin D$, 
so it meets $\bd D$, and thus
also $C$, at some point, say $c$.
Hence, $
 \dist(a,z)
= 
 \dist(a,c)+\dist(c,z)
>
 \dist(a, c)
\geq
 \dist(a, C)
$. 
We also have
\begin{align}
\label{equation: dom and bisect look the same from outside}
 \dist (a, C) & = \dist (a, D) & \text{for all} \ a \notin D. 
\end{align}
For let $d \in D$ be arbitrary. 
Again, there is a segment connecting $a$ and $d$, 
and it meets $\bd D$, and thus 
also $C$, at some point, say $c$. 
Hence, $
 \dist(a, d)
= 
 \dist(a, c) + \dist(c, d)
\geq
 \dist(a, c)
\geq
 \dist(a, C)
$.  Since $C \subseteq D$, this proves
\eqref{equation: dom and bisect look the same from outside}. 

The first part
of \ref{enumi: dominance region of the boundary} 
comes as follows: 
Points $a \in D$ belong both to $\dom (D, Z)$ and, 
by \eqref{equation: dom* and bisect look the same from inside}, 
to $\dom (C, Z)$; 
other points $a \notin D$ belong to $\dom (D, Z)$ 
and $\dom (C, Z)$ at the same time
by \eqref{equation: dom and bisect look the same from outside}. 

The second part is similar: 
Points $a \in D$ belong neither to $\dom (Z, D)$ nor 
to $\dom (Z, C)$ by \eqref{equation: dom* and bisect look the same from inside}; 
other points $a \notin D$ belong to $\dom (Z, D)$ and $\dom (Z, C)$ at the same time
by \eqref{equation: dom and bisect look the same from outside}. 
\end{proof}

Now we proceed with $k$-gradations.
Let $(R_i,S_i)_i$ be a $k$-gradation for $P$ and $Q$
as in Proposition~\ref{p:k-sector is a boundary}.
We observe that  
$R _i \cup S _i$ is the whole space and that
\begin{align}
\label{equation: chain of inclusions}
& P = R _0 \subseteq R _1 \subseteq \dots \subseteq R _{k - 1}, 
& 
& S _1 \supseteq S _2 \supseteq \dots \supseteq S _k = Q,
\end{align}
because $X \subseteq \dom (X, Y)$. 

\begin{lemma}
\label{lemma: separated} 
Let $P$, $Q$ be nonempty, disjoint, closed sets in
an arbitrary metric space. 
\begin{enumerate}
\renewcommand{\theenumi}{\textup{(\roman{enumi})}}
\renewcommand{\labelenumi}{\theenumi}
\item \label{enumi: k-sector curves are separated 2}
If $(C _i) _i$ is a $k$-sector 
of $C _0 = P$ and $C _k = Q$, then 
$C _{i - 1}$ and $C _{i + 1}$ are disjoint 
for each $i = 1$, \ldots, $k - 1$. 
\item \label{enumi: gradation sets are separated}
If $(R _i, S _i) _i$ is a $k$-gradation
between $R _0 = P$ and $S _k = Q$, then
$R _i$ and $S _j$ are disjoint for each $i$ and $j$ 
with $0 \leq i < j \leq k$. 
\end{enumerate}
\end{lemma}

\begin{proof}
Suppose that there is a point $a \in C _{i - 1} \cap C _{i + 1}$. 
Since $\dist (a, C _{i - 1}) = 0 = \dist (a, C _{i + 1})$, 
we have $a \in \bisect (C _{i - 1}, C _{i + 1}) = C _i$.  
Since $P$ and $Q$ are disjoint, 
either $a \notin P$ or $a \notin Q$. 
By symmetry, we may assume $a \notin P$. 
Let $i ^-$ be the smallest such that 
$a \in C _j$ for all $j = i ^-$, \ldots, $i$. 
Then $a \in C _{i ^- + 1} \setminus C _{i ^- - 1}$, 
contradicting
$a \in C _{i ^-} = \bisect (C _{i ^- - 1}, C _{i ^- + 1})$. 

For \ref{enumi: gradation sets are separated}, 
suppose that there is a point $a \in R _i \cap S _j$ for some $i < j$. 
Since $P$ and $Q$ are disjoint, 
either $a \notin P$ or $a \notin Q$. 
By symmetry, we may assume $a \notin P$. 
Retake $i$ to be the smallest such that $a \in R _i$. 
Then 
$a \notin R _{i - 1}$ and 
$a \in S _j \subseteq S _{i + 1}$, 
contradicting $
a \in R _i = \dom (R _{i - 1}, S _{i + 1})
$.  
\end{proof}

\begin{proof}[Proof of Proposition~\ref{p:k-sector is a boundary}.]
For one direction, 
let $(R _i, S _i) _i$ be 
a $k$-gradation and let $C _i = R _i \cap S _i$ 
for each $i = 1$, \ldots, $k - 1$. 
Then $
  C _i 
 =
  \dom(R _{i - 1}, S _{i + 1}) \cap \dom(S _{i + 1},R _{i - 1})
 =
  \bisect (R _{i - 1}, S _{i + 1})
$ is nonempty. 
Moreover, this equals $
 \bisect (C _{i - 1}, C _{i + 1})
$ by Lemma~\ref{lemma: properties of dom}\ref{enumi: bisector of the boundary}, 
because $R _{i - 1}$ and $S _{i + 1}$ are disjoint according
to Lemma~\ref{lemma: separated}\ref{enumi: gradation sets are separated}.

For the other direction, 
we suppose that $(C _i) _i$ is a $k$-sector. 
Let $R _i = \dom (C _{i - 1}, C _{i + 1})$ 
and $S _i = \dom (C _{i + 1}, C _{i - 1})$ 
for each $i = 1$, \dots, $k - 1$.
Then $C_i=R_i\cap S_i$ by the definition of a $k$-sector. 
By Lemma~\ref{lemma: separated}\ref{enumi: k-sector curves are separated 2}, 
we have
$R _i\cap C _{i + 1}=\emptyset$, 
and similarly $S _{i + 1}\cap C _i=\emptyset$. 
Therefore,  $R _i \cap S _{i + 1}$ is disjoint from $
 C _i \cup C _{i + 1} 
\supseteq 
 \bd R _i \cup \bd S _{i + 1} 
\supseteq
 \bd (R _i \cap S _{i + 1})
$.  This means that $R_i \cap S _{i + 1}$ has an empty boundary, 
and thus is itself empty, because the whole space is geodesic and hence connected. 
By this and the fact that $R _i \cup S _i$ covers the whole space, 
we have $P\subseteq R_1\subseteq\cdots\subseteq
R_{k-1}$ and 
$S_1\supseteq S_2\supseteq\cdots\supseteq S_{k-1}\supseteq Q$.
Because $R _i$ and $S _{i + 1}$ are disjoint, 
so are $R _{i - 1}$ and $S _{i + 1}$. 
This allows us to apply Lemma~\ref{lemma: properties of dom}\ref{enumi: dominance region of the boundary},
which yields $
 \dom (R _{i - 1}, S _{i + 1})
=
 \dom (C _{i - 1}, C _{i + 1})
=
 R _i
$ and similarly $
 \dom (S _{i + 1}, R _{i - 1})
=
 S _i
$. 
\end{proof}

The following example shows that 
the assumption of the space being geodesic cannot be dropped. 
Consider the distance on $\Rset$ defined by
$\dist (x, y) = f (\lvert x - y \rvert)$,
where 
$f$ is given by 
\begin{equation}
f (r) =
\begin{cases}
 r & \text{if} \ r \leq 1, \\
 1 & \text{if} \ 1 \leq r \leq 2, \\
 r/2 & \text{if} \ r \geq 2.
\end{cases}
\end{equation}
Thus, $d$ is almost like the usual metric, 
except that it ``thinks of any distance between $1$ and $2$ as the same.''
Then there is no trisector between
$P = (-\infty, 0]$ and $Q = [1, +\infty)$ 
(whereas there is a gradation by Proposition~\ref{p:gradations}). 
For suppose that $(C _1, C _2)$ is a trisector. 
By Lemma~\ref{lemma: separated}\ref{enumi: k-sector curves are separated 2}, 
the set~$C _2$ cannot overlap $P$ or $Q$, 
so it is a nonempty subset of $(0, 1)$. 
Hence, the point~$2$ is equidistant from $C _2$ and $P$, 
and thus belongs to $C _1$.  
This contradicts 
Lemma~\ref{lemma: separated}\ref{enumi: k-sector curves are separated 2}.

\section{\boldmath Drawing $k$-sectors}
\label{section: constructive}

Here we provide a more constructive proof of the existence
of $k$-gradations,  but only under stronger assumptions than in
Proposition~\ref{p:gradations}. 
Later we discuss how this
approach can be used for approximate computation of bisectors.
We write $\closure X$ for the closure of a set~$X$. 

\begin{prop}
\label{p:it-fix}
Suppose that $P$ and $Q$ are disjoint nonempty closed sets in
$\Rset ^d$ with the Euclidean norm (or, more generally,
with an arbitrary strictly convex norm).
Let the lattice $\mathcal L$ and 
the function $F \colon \mathcal L \to \mathcal L$ be as in the
the proof of Proposition~\ref{p:gradations} (Section~\ref{section: main proof}). 
Let $(R ^0 _i, S ^0 _i) _i$ be an arbitrary element of $\mathcal L$ with
$(R ^0 _i, S ^0 _i) _i \leq F((R ^0 _i, S ^0 _i) _i)$. 
Define $(R ^{n + 1} _i, S ^{n + 1} _i) _i := F ((R ^n _i, S ^n _i) _i)$ 
for each $n \in \Nset$ (thus, $
(R ^0 _i, S ^0 _i) _i \leq (R ^1 _i, S ^1 _i) _i \leq (R ^2 _i, S ^2 _i) _i \leq \cdots
$), and let $
 (R ^\infty _i, S ^\infty _i) _i 
=
 \bigvee \{\, (R ^n _i, S ^n _i) _i : n \in \Nset \,\}
$. 
Then $(\closure{R ^\infty _i}, S ^\infty _i) _i$ is a $k$-gradation.
\end{prop}

We begin proving this proposition. 
We write $R ^n _0 = P$ and $S ^n _k = Q$ for each $n \in \Nset \cup \{\infty\}$. 

\begin{lemma}
\label{lemma: bisector is the boundary of the dominance region}
For any disjoint nonempty closed sets $X$, $Y$ in $\Rset ^d$ 
with the Euclidean metric (or with a strictly convex norm), $
\dom (Y, X) = \closure{\Rset ^d \setminus \dom (X, Y)}
$. 
\end{lemma}

\begin{figure}
\begin{center}
\includegraphics[clip,scale=1.0]{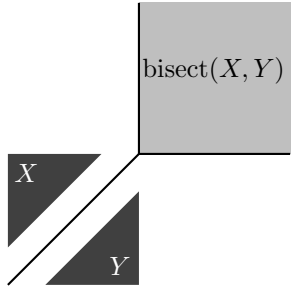}
\end{center}
\caption{A bisector may be fat
         in the plane with the $\ell_1$ metric. 
         Every point in the shaded region
         is at the same distance from $X$ and $Y$.
         The equation in
         Lemma~\ref{lemma: bisector is the boundary of the dominance region}
         does not hold.}
\label{figure: fat bisector}
\end{figure}

We note that the assumption on the considered metric in this lemma
is necessary: 
As Figure~\ref{figure: fat bisector} illustrates,
the claim is not valid  with the $\ell_1$ norm.

\begin{proof}[Proof of Lemma~\ref{lemma: bisector is the boundary of the dominance region}]
We have $\dom (Y, X) \supseteq \closure{\Rset ^d \setminus \dom (X, Y)}$
because 
$\dom (Y, X)$ is closed and $\dom (Y, X) \cup \dom (X, Y) = \Rset ^d$. 
For the other inclusion, 
let $z \in \dom (Y, X)$ and 
let $y$ be a closest point in $Y$ to $z$. 
Since $X$ does not intersect
the open ball with centre~$z$ and radius~$\dist(y,z)$, 
any point $z' \neq z$ on the segment $z y$ 
is strictly closer to $y$ than to $X$ 
(Figure~\ref{figure: dom and bisect}), and thus 
is not in $\dom (X, Y)$. 
Since $z'$ can be arbitrarily close to $z$, 
we have $z \in \closure{\Rset ^d \setminus \dom (X, Y)}$. 
\begin{figure}
\begin{center}
\includegraphics[clip,scale=1.0]{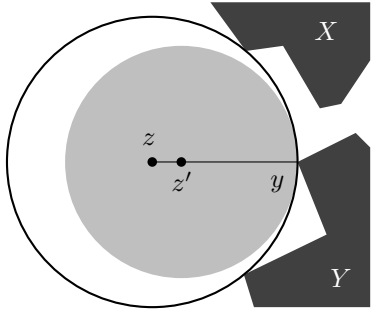}
\end{center}
\caption{Since $X$ does not intersect the interior of the ball around $z$, 
it does not touch the ball around $z'$.}
\label{figure: dom and bisect}
\end{figure}
\end{proof}

\begin{lemma}\label{eq12}
If $(R ^\infty _i, S ^\infty _i) _i$ is as in Proposition~\ref{p:it-fix},
then $\closure{R ^\infty _i} \cap S ^\infty _j=\emptyset$
whenever $0 \leq i < j \leq k$. 
\end{lemma}

\begin{proof}
For contradiction, suppose that there is some $
a \in \closure{R ^\infty _i} \cap S ^\infty _j 
$. 

If $i > 0$, then for each $n \in \Nset$ 
we have $
a \in S ^\infty _j \subseteq S ^n _j \subseteq S ^n _{i + 1}
$, so $
 \dom (R ^n _{i - 1}, \{a\})
\supseteq
 \dom (R ^n _{i - 1}, S ^n _{i + 1})
=
 R ^{n + 1} _i
$. 
This implies $
 \dist (a, R ^n _{i - 1})
\leq
 2 \cdot \dist (a, R ^{n + 1} _i)
$. 
Since $a \in \closure{R ^\infty _i}$, 
the right-hand side tends to $0$ as $n \to \infty$, 
and hence, so does $\dist (a, R ^n _{i - 1})$. 
Thus, $a \in \closure{R ^\infty _{i - 1}}$.
Repeating the same argument for $i-1$, $i-2$, \ldots, we arrive at
$a \in \closure{R ^\infty _0} = P$. 

Similarly, if $j < k$, then $
a \in S ^\infty _j \subseteq S ^{n + 1} _j = \dom (S ^n _{j + 1}, R ^n _{j - 1})
$ for all $n \in \Nset$. 
Thus, $
 \dist (a, S ^n _{j + 1})
\leq
 \dist (a, R ^n _{j - 1}) 
\leq 
 \dist (a, R ^n _i) 
\to 
 0
$ as $n \to \infty$ because $
 a \in \closure{R ^\infty _i}
$.  So $a \in S ^\infty _{j + 1}$.  
Repeating the argument
for $j + 1$, $j + 2$, \ldots, we obtain $a \in S ^\infty _k = Q$. 

Thus we have $a \in P \cap Q$, 
contradicting the assumption that $P$ and $Q$ are disjoint. 
\end{proof}

\begin{proof}[Proof of Proposition~\ref{p:it-fix}.]
Our goal is to show that $
F ((\closure{R ^\infty _i}, S ^\infty _i) _i) = (\closure{R ^\infty _i}, S ^\infty _i) _i
$.  Since $F$ is monotone, 
$
  F ((\closure{R ^\infty _i}, S ^\infty _i) _i) 
 \geq 
  F ((R ^n _i, S ^n _i) _i)
 \geq 
  (R ^n _i, S ^n _i) _i
$ for each $n$, and hence $
F ((\closure{R ^\infty _i}, S ^\infty _i) _i) \geq (\closure{R ^\infty _i}, S ^\infty _i) _i
$.  It remains to show that $
F ((\closure{R ^\infty _i}, S ^\infty _i) _i) \leq (\closure{R ^\infty _i}, S ^\infty _i) _i
$, which means,
by the definition of $F$, that 
\begin{equation}\label{e:first}
\dom (S ^\infty _{i + 1}, \closure{R ^\infty _{i - 1}})
\supseteq
 S ^\infty _i, 
\end{equation}
 and 
\begin{equation}\label{e:second}
\dom (\closure{R ^\infty _{i - 1}}, S ^\infty _{i + 1})
\subseteq
 \closure{R ^\infty _i}. 
\end{equation}

The inclusion~\eqref{e:first} follows just by continuity of the distance function:
We have  $S ^\infty _i = \bigcap _{n \in \Nset} S ^{n + 1} _i=
\bigcap _{n \in \Nset} \dom \bigl( S ^n _{i + 1}, R ^n _{i - 1} \bigr)$.
So for $x\in S ^\infty _i$
we have $\dist (x, S ^n _{i + 1}) \leq \dist (x, R ^n _{i - 1})$
for every $n$, and 
$\dist(x,S^\infty_{i+1})=\lim _{n \to \infty} \dist (x, S ^n _{i + 1}) \leq 
\lim _{n\to\infty} \dist (x, R ^n _{i - 1})=
\dist(x, \closure{R ^\infty _{i - 1}})$. 
Hence $x\in \dom (S^\infty _{i + 1}, \closure{R ^\infty _{i - 1}})$ and \eqref{e:first} is proved.

For proving \eqref{e:second}, we need the previous lemmas.
By \eqref{e:first}, we have 
\begin{equation}\label{e:second'}
 \Rset ^d \setminus \dom (S ^\infty _{i + 1}, \closure{R ^\infty _{i - 1}})
\subseteq
 \Rset ^d \setminus S ^\infty _i
\subseteq
 R ^\infty _i, 
\end{equation} 
where the latter inclusion is because 
$R ^n _i \cup S ^n _i = \Rset ^d$ 
for every $n$ 
(this was part of the definition of $\mathcal L$). 
We obtain \eqref{e:second} by 
taking the closure of \eqref{e:second'}, 
using Lemma~\ref{lemma: bisector is the boundary of the dominance region} 
for the left-hand side; 
for applying this lemma, 
we need $
\closure{R ^\infty _{i - 1}} \cap S ^\infty _{i + 1} = \emptyset
$, which holds by Lemma~\ref{eq12}. 
\end{proof}

If the initial element $(R ^0 _i, S ^0 _i) _i$ in Proposition~\ref{p:it-fix}
is less than or equal to all $k$-gradations 
(with respect to the ordering $\mathord\leq$), 
then so is $(R ^n _i, S ^n _i) _i$ for all $n$ (inductively 
by the monotonicity of $F$), 
and therefore, the resulting $
(\closure{R ^\infty _i}, S ^\infty _i) _i
$ is the \emph{least} $k$-gradation. 
This is the case when, for example, 
$(R ^0 _i, S ^0 _i) _i$ is  the least element $(P, \Rset ^d) _i$ of $\mathcal L$.  

The trisector in Figure~\ref{figure: l1 non-unique trisector}
corresponds to the least $3$-gradation, but this $3$-gradation
is not obtained by iteration from the least element of $\mathcal L$.
This witnesses that Proposition~\ref{p:it-fix} may indeed fail
for norms that are not strictly convex.

\paragraph{Computational issues}

Proposition~\ref{p:it-fix} gives 
a method to draw a $k$-sector in Euclidean spaces: 
By applying $F$ iteratively, 
we get an ascending chain $
(R ^0 _i, S ^0 _i) _i \leq (R ^1 _i, S ^1 _i) _i \leq \cdots
$ whose supremum $(R ^\infty _i, S ^\infty _i) _i$ gives 
a $k$-gradation $(\closure{R ^\infty _i}, S ^\infty _i) _i$. 
If we stop the iteration 
after sufficiently many steps, 
we obtain an approximation of $(\closure{R ^\infty _i}, S ^\infty _i) _i$.

However, implementing the algorithm is not entirely trivial, 
because even if the sites are simple, 
applying $F$ repeatedly 
gives rise to regions that are hard to describe. 
For example, 
consider the case where $P$ and $Q$ are points in the plane, 
and we begin with $(R ^0 _i, S ^0 _i) _i = (P, \Rset ^2) _i$. 
Then $\bd R ^1 _{k - 1}$ is the line bisecting $P$ and $Q$, 
and $\bd R ^2 _{k - 2}$ is the parabola bisecting $P$ and this line. 
The next iteration yields the curve $\bd R ^3 _{k - 3}$ (or $\bd R ^3 _{k - 1}$) 
which bisects between a parabola and a point. 

Thus, 
unlike typical basic operations allowed 
in computational geometry, 
taking the bisector gives rise to increasingly complicated curves. 
If we have an analytic description of the boundary curves of 
the regions $R ^n _i$ and $S ^n _i$, 
each of the curves defining $R ^{n + 1} _i$ and $S ^{n + 1} _i$ is described by
a system of differential equations associated with the bisecting condition.
But solving such equations exactly in each iterative step
is computationally expensive.  
Therefore, we need to find a practical method for 
approximating the bisectors (assuming that we only compute 
the regions in a bounded area). 

One method is to approximate each region by 
a polygonal region. 
We start with some polygonal approximations
$\tilde P$, $\tilde Q$ of $P$, $Q$, 
and let $(\tilde R ^0 _i, \tilde S ^0 _i) _i := (\tilde P, \Rset ^d) _i$. 
Then for each $n$, 
we compute an approximation $(\tilde R ^{n + 1} _i, \tilde S ^{n + 1} _i) _i$ 
to $F ((\tilde R ^n _i, \tilde S ^n _i) _i)$, 
where the bisector of two polygonal regions, 
which is a piecewise quadratic curve, 
is approximated  by a suitable polygonal region. 
To ensure that $(\tilde R ^n _i, \tilde S ^n _i) _i$ converges to an underestimate
(with respect to the ordering $\leq$)
of the least $k$-gradation $(\closure{R ^\infty _i}, S ^\infty _i) _i$, 
we should have $
(\tilde R ^n _i, \tilde S ^n _i) _i \leq (\tilde R ^{n + 1} _i, \tilde S ^{n + 1} _i) _i \leq F ((\tilde R ^n _i, \tilde S ^n _i) _i)
$. 
This can be achieved by computing 
an inner approximation of $R _i ^{n + 1}$ and 
an outer approximation of $S _i ^{n + 1}$. 

Another method is to consider the problem in the pixel geometry,
where each of the approximate regions $\tilde R ^n _i$, $\tilde S ^n _i$ 
is a set of pixels. 
In computing $(\tilde R ^{n + 1} _i, \tilde S ^{n + 1} _i) _i$, 
we again make sure that $
(\tilde R ^n _i, \tilde S ^n _i) _i \leq (\tilde R ^{n + 1} _i, \tilde S ^{n + 1} _i) _i \leq F ((\tilde R ^n _i, \tilde S ^n _i) _i)
$.  Then $(\tilde R ^n _i, \tilde S ^n _i) _i$ stabilizes eventually, 
providing a lower estimate of the least $k$-gradation.
The analysis of time complexity (as a function of precision) 
of these methods is left as a future research problem. 

\paragraph{Uniqueness}

The curves in Figure~\ref{figure: k-sector example} were drawn using
the pixel geometry model explained above. 
As we mentioned there, they are guaranteed to
lie on $P$'s side of any true $k$-sector curves. 
By exchanging $P$ and $Q$, 
we obtain also an approximate $k$-sector that lies
on $Q$'s side of any true $k$-sector. 
We tried computing these lower and upper estimates for 
several different $P$, $Q$ and $k$ in Euclidean plane, 
but we did not find them differ by a significant amount. 
Because of this, we suspect that the $k$-sector is always unique:

\begin{conjecture}
The $k$-sector of any two disjoint nonempty closed sets in Euclidean space 
is unique. 
\end{conjecture}

\subsection*{Acknowledgements}

We gratefully acknowledge valuable discussions with many friends 
including Tetsuo Asano and G\"unter Rote;
indeed, we owe Tetsuo for precious information of his 
recent work on convex distance cases.
We also thank Yu~Mura\-matsu for his programming work in drawing figures.
D.\,R. would like to express his thanks to 
Simeon Reich for his helpful discussion.
Finally, we remark that the warm comments from the 
audience of the preliminary announcement at EuroCG~2009 
encouraged us to work further on the subject. 

A.\,K. is supported by
the Nakajima Foundation and 
the Natural Sciences and Engineering Research Council of Canada. 
The part of this research by T.\,T.
was partially supported by the 
JSPS Grant-in-Aid for Scientific Research (B) 18300001.

\end{document}